\documentclass[manyauthors]{fundam}
\usepackage{url} 
\usepackage[ruled,lined]{algorithm2e}
\usepackage{graphicx}
\usepackage{tikz}
\usetikzlibrary{automata, positioning, arrows}

\usepackage{amsfonts}
\usepackage{amssymb}

\usepackage{xcolor}
\usepackage{graphicx}

\newtheorem{problem}{Problem}

\usepackage{xspace}

\newcount\Comments  
\definecolor{darkgreen}{rgb}{0,0.6,0}
\newcommand{\kibitz}[2]{\ifnum\Comments=1{\color{#1}{#2}}\fi}

\newcommand{\fcal}{\ensuremath{\mathcal{F}}\xspace}
\newcommand{\kcr}{{\sc $k$-charge removal}\xspace}
\newcommand{\akcr}{{\sc $k_{\geq}$-charge removal}\xspace}
\newcommand{\minkcr}{{\sc minimal-$k_{\geq}$-charge removal}\xspace}

   \usepackage{hyperref}

\begin{document}

\setcounter{page}{181}
\publyear{2021}
\papernumber{2096}
\volume{184}
\issue{3}

   \finalVersionForARXIV

\title{On the Hardness of Energy Minimisation \\for Crystal Structure Prediction\footnote{A preliminary conference version of this work appeared in Proceedings of the 46th International Conference on Current Trends in Theory and Practice of Computer Science,
SOFSEM 2020 \cite{Adamson2020}.}}

\author{Duncan Adamson\thanks{Address  for correspondence: Department of Computer Science,
               University of Liverpool, Liverpool, L69 3BX, UK.
               \newline \newline
          \vspace*{-6mm}{\scriptsize{Received October 2021; \ revised December 2021.}}}
\\
Department of Computer Science\\
University of Liverpool, Liverpool,  UK \\
duncan.adamson@liverpool.ac.uk
\and
Argyrios Deligkas\\
Department of Computer Science\\
Royal Holloway University of London,  London, UK\\
Argyrios.Deligkas@rhul.ac.uk
\and
Vladimir Gusev\\
Leverhulme Research Centre for Functional \\
Materials Design\\
University of Liverpool, Liverpool, UK\\
Vladimir.Gusev@liverpool.ac.uk
\and
Igor Potapov\\
Department of Computer Science\\
 University of Liverpool, Liverpool, UK\\
 potapov@liverpool.ac.uk
}

\maketitle

\runninghead{D. Adamson et al.}{On the Hardness of Energy Minimisation for Crystal Structure Prediction}

\vspace*{-5mm}
\begin{abstract}
Crystal Structure Prediction (\textsc{csp}) is
one of the central and most challenging problems in materials science and computational chemistry.
In \textsc{csp}, the goal is to find a configuration of ions in 3D space that yields the lowest
potential energy.
Finding an efficient procedure to solve this complex optimisation question is a well known open problem. 
Due to the exponentially large search space, the problem has been referred in several materials-science papers as ``NP-Hard and very challenging'' without a formal proof.
This paper fills a gap in the literature providing the first set of formally proven
NP-Hardness results for a variant of \textsc{csp} with various realistic constraints.
In particular, we focus on the problem of \emph{removal}: the goal is to find a
substructure with minimal potential energy, by removing a subset of the ions.
Our main contributions are NP-Hardness results for the \textsc{csp} \emph{removal} problem, new embeddings of combinatorial graph problems into geometrical settings, and a more systematic exploration of the energy function
to reveal the complexity of \textsc{csp}.
In a wider context, our results contribute to the
analysis of computational problems for weighted graphs embedded into the three-dimensional Euclidean space.\vspace*{-1mm}
\end{abstract}



\section{Introduction}

One of the central and most challenging problems in materials science and computational chemistry is to predict the structure of a crystal, given the set of ions composing it.
At a high level, the goal there is to find the ``best'' configuration structure of ions in a three-dimensional box.
More specifically, the input of the problem consists of a chemical formula, and a function that maps configurations of the ions -- that agree with the chemical formula -- in the Euclidean three-dimensional space to energy.
The objective is to find a structure that minimises the average energy per ion, since this structure is more likely to correspond to a stable material \cite{LYAKHOV20101623}.
This problem, termed \emph{Crystal Structure Prediction} (\textsc{csp}), has remained open due to the complexity of solving it optimally \cite{woodley2008crystal} and the combinatorial explosion following a brute-force approach.
There are many previous approaches to this problem, largely based on heuristic techniques \cite{csp-local,LYAKHOV20101623, doi:10.1002/1521-3765(20020916)8:18<4102::AID-CHEM4102>3.0.CO;2-3,oganov2006crystal,WANG20122063}, however these lack the ability to guarantee optimality within any factor and moreover they are computationally very demanding.

In the most generic formulation of \textsc{csp} there are many degrees of freedom,
due to both the physical constraints of the setting, such as the interaction between types of ions, and the optimisation parameters of the model such as the number of ions to place, and their positions.
Furthermore, real crystals are based on periodic tessellations of 3D space with \emph{unit cells}, parallelepiped periods composed of ions, whose size and shape may also be changed.
Even when the size with respect to both the size of the unit cell, and the number of ions are restricted, the search space for this problem is still exponential.
Due to this, \textsc{csp} has, incorrectly, been referred to in several computational-chemistry papers as ``NP-Hard and very challenging'' \cite{Oganov2018}.
However, from the computational-theory viewpoint the argument that the search must be done in a set of exponential size does not imply NP-Hardness.

Two results closely related to the NP-Hardness of energy minimisation can be found in \cite{Barahona_1982} and in \cite{Wille_1985}.
In \cite{Barahona_1982}, within the context of the Ising model, the authors show NP-Hardness for placing $\pm{1}$ charged vertices on a graph taking into account only interactions between connected vertices.
The reduction works on a grid, where each vertex has degree at most 6, making the interaction very local.
While this is close to \textsc{csp}, the limitation to only very local interactions within the context of the electrostatic charge reduces the relevance of this result to \textsc{csp}.
In \cite{Wille_1985}, it was shown that the problem of placing ions for some given positions is in NP. However, the reduction goes only one way, and thus shows only containment in NP, and does not imply the NP-Hardness of the problem.
While both of these results are closely related, neither of them provides a satisfying answer to the 
\textsc{csp} problem.

In our work, we consider several special variants of \textsc{csp} and provide a few alternative reasons for the hardness of closely related problems.
We take inspiration from hard combinatorial problems in graph theory and embed several NP-Hard problems on graphs into complete weighted graphs in Euclidean space.
These can be seen as optimisation problems for weighted geometric graphs with a non-linear objective function.
We focus on the problem of \emph{removal};
which can be seen as a combinatorial variant of \textsc{csp}. Here, the input is a configuration of  ions,
and the goal is to remove a subset of them such that the \emph{interaction} energy -- the pairwise energy between each pair of ions --  among the remaining atoms is minimised. This can be seen as a problem of removing vertices of a graph which results in a subgraph satisfying some specific property.
These types of problems have been intensively studied in the combinatorial graph theory.
In \cite{doi:10.1137/0208049}, it was shown that is is NP-Complete to decide if a removal of $k$ vertices from a graph leaves a forest, or planar graph.
In \cite{doi:10.1137/0210022} and  \cite{Yannakakis78node-and} this was extended to further properties showing NP-Completeness for bipartite graphs and for non-trivial hereditary properties.

An alternative view of the removal problem is as variant of \textsc{csp}, where the available positions of the ions correspond to points on a discrete grid/lattice of 3D Euclidean space.
For example, we can find the optimal structure for an instance of \textsc{csp}  as follows. Firstly, we can unrealistically place several ions from which we want to build a new structure at every position of the discrete space (or even at unrealistically close distance to each other at every position).
Then we need to find a set of ions that should be removed to minimise the energy function.
Due to the nature of the energy function, when the goal is to minimise the potential energy, the overlapping ions must be removed and thus we get a realistic structure.
In our variant of \emph{removal} problem, for which we show NP-Hardness results, an initial configuration of ions on this grid, from which we remove ions, is part of the input and has only vacant positions or positions with a single ions in the discrete three-dimensional-Euclidean space.

In this paper the interaction is restricted with respect to the energy function to a single unit cell only.
The primary reason for this restriction is that the electrostatic potential over the infinite series when expressed naively converges conditionally \cite{Faber2015}.
While there are approaches that mitigate this issue, they are difficult to analyse theoretically.
An alternative approach, used here, is to consider a sufficiently large part of a crystal as a unit cell
such that the local interactions within the cell provide a good approximation of the total energy.
By showing that this problem is hard even in this simple case,
we provide a strong evidence that the general case is hard as well.

\paragraph{\textbf{Our contributions.}}
Our main goal is to initiate the study of
\textsc{csp} through the lens of Theoretical Computer Science and explore a more systematic way of studying the energy function that could reveal the computational complexity of \textsc{csp}.
Towards this, we give the first correct NP-Hardness results with more realistic constraints for \textsc{csp} such as those inspired by \cite{collins2017accelerated}.
In addition, we provide new embeddings of combinatorial problems on graphs in geometrical settings.
Our results can be seen as part of a more general problem of removing vertices from a weighted graph embedded into 3D Euclidean space.
This type of results are more difficult to obtain compared to other metric spaces.

There are two main technical challenges in relation to the graphs that we consider.
Firstly, they are {\em complete}. Secondly, they are ``Euclidean'', i.e., the edges are weighted proportional to the Euclidean distance between their vertices.
In this restricted setting, many classical NP-Hard problems are much harder to embed.
There are only a few NP-hardness results about low dimensional Euclidean graphs \cite{ITA_2011__45_3_331_0, CLARK1990165}.
The main difficulty is to satisfy the triangle inequality in constant dimensions.
This is further complicated in our setting by the additional weighting caused by a non-linear energy function acting as a weight on the edges.

\medskip
We consider three versions of the \emph{removal} problem.
In every version, the input consists of a graph $G$ embedded into 3D Euclidean space and an ``energy function'' that maps every subgraph of $G$ to an energy.
Further each vertex in $G$ is labelled with an ion species and it is weighted according to the charge of the species.
Similarly each edge $e$ in $G$ is weighted using the energy function on the subgraph consisting of only $e$.
The objective for each version is to remove a \emph{neutral} subset of vertices - a set of vertices for which the sum of charges is zero - such that the energy of the remaining subgraph is~minimised.~Informally speaking the versions differ by constraints imposed on the set of vertices we can remove.
\begin{itemize}
\item {\bf $k$-Charge Removal.}
Here, we must remove a neutral set of vertices such that the sum of vertices with positive charges is equal to $k$.
In Corollary \ref{col:np_for_sol_kcr} we note that it is NP-Hard to determine if there is a solution to this problem when the charges on the vertices are unbounded.
\item {\bf $k_\geq$-Charge Removal}.
Here, a neutral set of vertices must be removed such that the sum of vertices with a positive charge is greater than or equal to $k$.
\item {\bf Minimal-$k_{\geq}$-Charge Removal.}
Here, the neutral set of vertices must also be minimal (informally, no vertices can be removed from the set while keeping it both neutral and the sum of positive charges greater than or equal to $k$).
\end{itemize}

\begin{table}[!b]
\vspace*{-3mm}
 \caption{Summary of our results and their corresponding settings.} \label{tab:summary}
    \scalebox{0.95}{
    \begin{tabular}{|p{2cm}|p{3.8cm}|p{8.3cm}|}
        \hline
        Theorem & Summary & Setting \\
        \hline
        Theorem \ref{thm:kcr-np} & NP-Completeness by reduction from the clique problem.& All problems, under any energy function in \fcal, vertices of charge $\pm c$ for a given $c$ and an unbounded number of ion species.\\
        \hline
        Theorem \ref{thm:kcr_np_any_charges} & NP-Completeness by extension of Theorem \ref{thm:kcr-np}. & All problems, under any energy function in \fcal, any bounded set of charges and an unbounded number of ion species.\\
        \hline
        Theorem \ref{thm:non_approximable} & Non-approximability in Polynomial time for any factor of $n^{1 - \epsilon}$, for $\epsilon > 0$ & \akcr, for any energy function in \fcal.\\
        \hline
        Theorem \ref{thm:kcr_max_weight_clique} & Reduction to max-weight-$k$-clique. & $k$-Charge Removal or minimal-$k_{\geq}$-Charge Removal under any computable energy function, vertices of charge $\pm c$ for a given $c$, and a unbounded number of ion species.\\
        \hline
        Theorem \ref{thm:kcr_2_np} & NP-Completeness by reduction from independent set on penny graphs. & All problems, under the Buckingham-Coulomb potential energy function, vertices of charge $\pm 1$, and two species of ion.\\
        \hline
        Theorem \ref{thm:kcr_coulomb_knapsack} & NP-Completeness by reduction from the knapsack problem. & Minimal-$k_{\geq}$-Charge Removal and $k_{\geq}$-Charge Removal, under the Coulomb potential energy function, unbounded number of vertex weights  and unbounded number of ion species.\\
        \hline
    \end{tabular} }
\end{table}

We also consider a variety of settings for these problems, varying the energy function for the crystal, restrictions on the charges of the vertices, and restrictions on the number of ion species.
We summarise our results in Table \ref{tab:summary}.
Regarding energy functions, we primarily consider the Buckingham-Coulomb potential function, which is commonly used in computational chemistry, and the more general class of \emph{controllable} energy functions, which we denote by \fcal, that includes the Buckingham-Coulomb potential function.
The class \fcal is formally defined in Section \ref{sec:prelim}.
We show in Proposition \ref{prop:bc_in_fcal} that Buckingham-Coulomb belongs to the class \fcal.
One other energy function we consider is the Coulomb potential, which is used to calculate the electrostatic potential.
We show that depending on the energy function used, and the restrictions on the ion species and vertex charges, we are able to reduce several different combinatorial problems to our problems.
Since the initial publication of these results Theorem \ref{thm:kcr_2_np} has been used to show the hardness of the related problem of \emph{Multimarginal Optimal Transport for Density Functional Theory} \cite{Altschuler2020,DO_ALTSCHULER2021}.

The remainder of this work is organised as follows: in Section \ref{sec:prelim} we discuss the preliminaries of these problems, providing relevant notation and definitions.
In Section \ref{sec:kcr} we present our results for the general case of the problems, proving NP-Hardness with Theorems \ref{thm:kcr-np} and \ref{thm:kcr_np_any_charges} for energy functions in \fcal and an unbounded number of ion species.
We also consider some natural restrictions to this problem.
In Section \ref{sec:kcr_ind_set} we consider the restriction of having only two species of ions under the \emph{Buckingham-Coulomb} potential as energy function. We show that the problem remains NP-Hard under these restrictions (Theorem \ref{thm:kcr_2_np}).
In Section \ref{sec:kcr_kanpsack} we consider the restriction to only the \emph{Coulomb} potential, this time with no restrictions on the number of species of ions or the charges of the ions.
In Theorem \ref{thm:kcr_coulomb_knapsack} we show that under these restrictions the problem remains NP-Hard.

\section{Preliminaries}
\label{sec:prelim}

\noindent
{\bf Unit Cell.}
A {\em crystal} is a solid material defined by an infinitely repeating parallelepiped period, called \emph{unit cell}.
A unit cell is a parallelepiped region of three dimensional space containing ions at fixed positions.
Each unit cell contains a set of $n$ \emph{ions} within the unit.
Each ion, $i$, has a \emph{specie}, e.g. Ti or Sr, and a non-zero charge $q_i$.
The specie for an ion $i$ is denoted $S(i)$.
In every crystal the unit cell is charge neutral, i.e., $\sum\limits_{1 \leq i \leq n} q_i = 0$.
An \emph{arrangement} of ions in a unit cell defines a position for every ion in the unit cell, i.e. the positions within $\mathbb{R}^3$.

\medskip \noindent
{\bf Energy of a Crystal.}
When two crystal structures of the same composition are assessed for stability,
a common method is to compare the sum of all pairwise interactions, with a more negative potential energy being preferable.
This energy is calculated using \emph{force fields}, determined
by the species of ions.
The force field parameters and distance between the ions are used by a given pairwise energy function $U$.
In general, energy is defined via series as a crystal is considered to have infinite size.

For notation, the pairwise interaction between two ions $i$ and $j$ with respect to the energy function $U$ is $U(i,j)$, denoted $U_{ij}$ when it is clear from context.
The value of $U_{ij}$ is defined by the force field of the ions and the Euclidean distance between them, which is included as one of the parameters.
The total potential energy for an arrangement of $n$ ions in the graph $G$ is given by $U(G) = \sum\limits_{1 \leq i,j \leq n, i \neq j} U_{ij}$.

This paper considers a general class of energy functions, called the \emph{controllable} potential functions, denoted by \fcal.
All functions in \fcal are required to be computable in polynomial time for any input.
Intuitively, for every $f \in \fcal$ there exists a set of force field parameters that counteract the effect of the distance parameter $r$.

\begin{definition}
A function $f: \mathbb{R}^q \mapsto \mathbb{R}$ belongs to \fcal if and only if for any given $a \in \mathbb{R}$ and any fixed $r \in \mathbb{R}^+$ there exists a vector of $q - 1$ parameters $\vec{x} \in \mathbb{R}^{q - 1}$ such that $f(\vec{x}, r) = a$.
\end{definition}

\noindent
One of the most common functions for crystal structure prediction is the \emph{Buckingham-Coulomb} potential  \cite{buckingham1938classical}, which is the sum of the Buckingham potential and the Coulomb potential.
The Coulomb potential for a pair of ions $i,j$ is defined as $U^C(i,j) =\frac{q_i q_j}{r_{ij}}$, where $r_{ij}$ is the Euclidean distance between the ions.
The Buckingham potential, $U^B(i,j)$, for a pair of ions $i,j$ is defined by four parameters: the distance between the ions, $r_{ij}$, and the three force field parameters, $A_{i,j}$, $B_{i,j}$, $C_{i,j} \in \mathbb{R}$, defined by the species of the ions.
Note that all three parameters are positive values.
The energy is calculated as $U^B(i,j) = \frac{A_{i,j}}{e^{B_{i,j} \cdot r_{ij}}} - \frac{C_{i,j}}{r_{ij}^6}$.
The {\bf Buckingham-Coulomb potential} is given by:
\begin{equation*}
    U^{BC}(i,j) = U^B(i,j) + U^C(i,j) = \frac{A_{i,j}}{e^{B_{i,j} \cdot r_{ij}}} - \frac{C_{i,j}}{r_{ij}^6} + \frac{q_i q_j}{r_{ij}}.
\end{equation*}

\begin{proposition}
    $U^{BC}: \mathbb{R}^4 \mapsto \mathbb{R}$ belongs to the class \fcal.
\label{prop:bc_in_fcal}
\end{proposition}

\begin{proof}
To show that $U^{BC}$ belongs to $\mathcal{F}$, it is sufficient to provide a constructive method to determine the force field parameters such that for any value $a \in \mathbb{R}$ and a pair of ions $i,j$ the energy $U^{BC}(i,j) = a$.
Let $i$ and $j$ be at a distance of $r_{ij}$ with arbitrary charges $q_i$ and $q_j$, the parameters may be set so that the potential at a distance of $r_{ij}$ is $a$.
We set $B_{i,j}=0$ and the values of $A_{i,j}$ and $C_{i,j}$ as follows depending on whether $a$ is positive or not. If $a > 0$, then we set:
\[
A_{i,j} = \begin{cases}
        a & \text{if } q_i q_j > 0;\\
        a + \frac{|q_i q_j|}{r_{ij}} & \text{otherwise}.
    \end{cases}\\ \qquad \text{and} \qquad
C_{i,j} = \begin{cases}
        q_i q_j r_{ij}^5 & \text{if } q_i q_j > 0;\\
        0 & \text{otherwise}.
    \end{cases}
\]
\noindent
If $a \leq 0$, then we set:
\[
A_{i,j} = \begin{cases}
        0 & \text{if } q_i q_j > 0;\\
        \frac{|q_i q_j|}{r_{ij}} & \text{otherwise}.
    \end{cases}\\ \qquad \text{and} \qquad
C_{i,j} = \begin{cases}
        |a| r_{ij}^6 + q_i q_j r_{ij}^5 & \text{if } q_i q_j > 0;\\
        |a| r_{ij}^6 & \text{otherwise}.
    \end{cases}
\]

\noindent
In this case of the Buckingham-Coulomb potential, the equation becomes
\begin{equation*}
    U^{BC}(i,j) = A_{i,j} - \frac{C_{i,j}}{r_{i,j}^6} + \frac{q_i q_j}{r_{ij}}.
\end{equation*}
\noindent
The Coulomb potential, $\frac{q_i q_j}{r_{ij}}$, is cancelled either by adding $q_i q_j r_{ij}^5$ to $C_{i,j}$, if $q_i q_j > 0$, or $\frac{|q_i q_j|}{r_{ij}}$ to $A_{i,j}$ in the case $q_i q_j \leq 0$.
In the first case the energy added by the Coulomb potential is $\frac{|q_i q_j|}{r_{ij}}$, which is cancelled by the addition of $q_i q_j r_{ij}^5$ when multiplied by the $\frac{-1}{r_{ij}^6}$ term applied to $C_{i,j}$.
Otherwise the Coulomb energy is $\frac{-|q_i q_j|}{r_{ij}}$, which is cancelled out by the relevant addition from $A_{i,j}$.
With the Coulomb energy removed, $a$ is either left as part of the $A_{ij}$ term, if $a > 0$, or part of the $C_{ij}$ term otherwise.
\end{proof}

\noindent
\textbf{Crystals as geometric graphs.}
The unit cell of a crystal may be thought of as a geometric graph $G = (V,E)$ in $\mathbb{R}^3$.
Recall that each ion corresponds to a charged point in $\mathbb{R}^3$.
In the context of a graph each ion represents a weighted vertex in $\mathbb{R}^3$ at the same position as the ion, giving a total of $n$ vertices.
\eject
\noindent Let the vertex $v_i$ correspond to ion $i$ and $wt(v_i)$ correspond to the charge of $q_i$, i.e. $wt(v_i) = q_i$.
For notation, $V^+ \subset V$ denotes the set of vertices with a positive charge, and $V^-$ the set of vertices with a negative charge.

\medskip
Between each pair of vertices there is an edge, weighted by the pairwise interaction of the corresponding ions $U_{ij}$.
Note that from its definition $U_{ij}$ is determined in part by the length of the edge, which is drawn as a straight line in the space.
The energy of a crystal graph $G = (V, E)$ is computed as $U = \sum\limits_{(v_i,v_j) \in E} U_{ij}$.

\begin{definition}
\label{def:crystal_graph}
A crystal graph is a complete geometric graph $G=(V,E)$ equipped  with a weight function $f: (v_1,v_2 \in V, r \in \mathbb{R}) \mapsto \mathbb{R}$ such that:
\begin{itemize}
    \item every $v \in V$ is associated to a non-zero integer $wt(v)$;
    \item $\sum_v wt(v)=0$;
    \item every edge $(u,v)$ is associated with a weight given by $f(u,v,dist(u,v))$, denoted $wt(u,v)$.
\end{itemize}
\end{definition}

In the remainder of this work crystals are described in terms of their physical structure, where it makes sense to be considering the ions, and as a graph otherwise.
Further, we assume that a crystal graph contains $n$ vertices unless otherwise stated, equivalently we assume that a crystal contains $n$ ions unless otherwise stated.
One important concept for crystals is the idea of \emph{neutrality}.
Informally a crystal is neutral if the sum of the charges of every ion in the unit cell is 0.

\begin{definition}
A set of vertices $R \subset V$ is \textbf{neutral} if $\sum\limits_{v_i \in R} wt(v_i) =  0$.
\label{def:neutral}
\end{definition}\vspace*{-3mm}

\paragraph*{\textbf{The \emph{k-Charge Removal} Problem.}}

The $k$-Charge Removal problem, henceforth \kcr, takes as input a crystal graph $G$ corresponding to a ``dense'' initial arrangement of ions, with the goal of removing some vertices in order to minimise the energy of the new subgraph $G' \subset G$.
It is assumed that the initial graph is \emph{neutral} by Definition \ref{def:neutral}.
As $G'$ must also be neutral, any set of vertices which is removed must therefore be neutral.
A natural number $k$ of charges to remove is chosen, as defined in Definitions \ref{def:k_charges} and \ref{def:removal}.
In practical applications, the value of $k$ may be chosen either using intuition from chemistry, or by exhaustively checking each value of $k$.

\begin{definition}
For any set $S \subseteq V,$ $S^+$ denotes the set of positively charged vertices in $S$, and $S^-$ the set of negatively charged vertices.
\end{definition}

\begin{definition}
A \textbf{set of k-charges} $R$ in a crystal graph $(V,E)$ is a neutral set of vertices, where $R \subseteq V$ and $\left\lvert \sum\limits_{v \in R^+} wt(v)\right\rvert =  \left\lvert \sum\limits_{v \in R^-} wt(v) \right\rvert = k$.
\label{def:k_charges}
\end{definition}

\noindent
Informally, a set of $k$-charges is a set of vertices with a net charge of 0, where the magnitudes of the sums of all positively charged vertices is $k$, as is the sum of all negatively charged vertices.

\begin{definition}
The {\bf removal} operation of a set of vertices $R$ from a graph $G = (V,E)$ returns the graph $G' = (V',E')$, where $V' = V \setminus R$ and $E'$ is the set of edges in $E$ with no endpoint in $R$.
\label{def:removal}
\end{definition}

\noindent
In other words, a removal of $R$ from $G$ returns the graph $G'$ that is the complement of the graph induced by $R$.

\begin{problem}
$k$-Charge Removal (\kcr)
\label{prob:kcr_def}
\end{problem}

\noindent
\begin{tabular}{l l}
    \emph{Instance:} & A crystal graph $G$, with edges weighted by a given common energy\\
    & function $U$, a natural number $k$ and a goal energy $g \in \mathbb{R}$.\\
    \emph{Question:} & Does there exist a set of $k$-Charges $R \subset V$
    such that removing $R$ from $G$\\
    & returns a graph $G'$, where $U(G') \leq g$?
\end{tabular}

\paragraph*{\textbf{The \emph{$k_{\geq}$-Charge Removal} Problem.}}

One variation of \kcr is the \emph{$k_{\geq}$-Charge Removal problem}, denoted \akcr.
This problem takes the same input as in \kcr, however rather than looking to remove a set of exactly $k$-charges, it is instead sufficient to remove a neutral set of at least $k$ positive and negative charges.
In this generalisation, vertices of total weight more than $k$ may be removed, provided the cell remains neutral.
Note that any removal of exactly $k$ is also acceptable for this generalisation.

\begin{definition}
A set of \textbf{$k_{\geq}$-charges} $R$ from a crystal graph $(V,E)$ is a neutral subset, where $R \subseteq V$ and $\sum\limits_{v_i \in R^+} wt(v_i) \geq k$.
\label{def:k_charge}
\end{definition}

\begin{problem}
$k_{\geq}$-Charge Removal (\akcr)
\end{problem}

\begin{tabular}{l l}
    \emph{Instance:} & A crystal graph $G$, with edges weighted by a given common energy\\
    & function $U$, a natural number $k$, and a goal energy $g \in \mathbb{R}$.\\
    \emph{Question:} & Does there exist a set of $k_{\geq}$-Charges $R \subset V$ such that removing $R$ \\
    & from $G$ returns a graph $G'$, where $U(G') \leq g$?
\end{tabular}

\begin{proposition}
A solution to \kcr or \akcr can be verified in polynomial time.
\label{prop:kcr_in_np}
\end{proposition}

\begin{proof}
A solution to \kcr contains the set of charges $R$ that are removed. 
This can be verified as a set of $k_{\geq}$-charges by simply summing up the positive and negative weights, checking that the set is neutral and that $\left\lvert\sum\limits_{v_i \in R^+} wt(v_i)\right\rvert = k$ for \kcr or $\left\lvert\sum\limits_{v_i \in R^+} wt(v_i)\right\rvert \geq k$ for \akcr.
The time complexity is of the order of $O(|R|)$.
Similarly the sum of the edges in the original graph $G$ that do not have an endpoint in $R$ can be checked against the goal value $g$.
This is done in $O(|V|^2)$ time, as the graph is complete.
As no step takes more than $O(|V|^2)$ time, a solution to either \kcr or \akcr can be verified in polynomial time.
Hence \kcr and \akcr fall into the class of NP problems.
\end{proof}

\paragraph*{\textbf{The \emph{Minimal-$k_{\geq}$-Charge Removal} Problem.}}
An alternative variation of \akcr is the \emph{minimal-$k_{\geq}$-Charge Removal} problem, denoted \minkcr.
This also serves as a generalisation of \kcr, where the goal is to get close to a set of $k$-charges, accepting that it may not be possible to reach the exact value.
In this problem a \emph{minimal} set of $k_{\geq}$-charges is removed.

\begin{definition}
A set $R$ of $k_{\geq}$-charges is \textbf{minimal} if there exists no subset $R' \subset R$ such that $\left\lvert \sum\limits_{v_i \in R'^+} wt(v_i) \right\rvert \geq k$ and $\left\lvert \sum\limits_{v_i \in R'^+} wt(v_i) \right\rvert = \left\lvert\sum\limits_{v_j \in R'^-} wt(v_j)\right\rvert$.
\label{def:min_k_charge}
\end{definition}

\noindent
Informally, Definition \ref{def:min_k_charge} means that there is no way of getting closer to a set of $k$-charges from the set, without having fewer than $k$ charges.
It follows that for a given crystal graph, there may be multiple minimal $k_{\geq}$-charge sets for a given $k$.
A removal of $k_{\geq}$-charges is minimal if the set of $k_{\geq}$-charges is minimal.
It may be noted that a set of $k$-charges is always a minimal set of $k_{\geq}$-charges.

\begin{problem}
Minimal-$k_{\geq}$-Charge Removal (\minkcr)
\label{prob:minkcr_def}
\end{problem}

\begin{tabular}{l l}
    \emph{Instance:} & A crystal graph $G$, with edges weighted by a given common energy\\
    & function $U$, a natural number $k$, and a goal energy $g \in \mathbb{R}$.\\
    \emph{Question:} & Does there exist a minimal set of $k_{\geq}$-charges $R \subset V$ such that\\
    & removing $R$ from $G$ returns a graph $G'$, where $U(G') \leq g$?
\end{tabular}

\begin{proposition}
It is NP-Hard to verify if a set of $k_{\geq}$-charges is minimal when no bounds are given on the charges of the vertices.
\label{prop:kcr_not_np}
\end{proposition}

\begin{proof}
This is shown by a reduction from the \emph{subset-sum} problem.
In the subset-sum problem there is a set of values $S$, and a goal $k$.
The task is to choose some subset $S' \subseteq S$ such that $\sum\limits_{i \in S'} i = k$.
Note that this problem remains NP-complete in the case the input is only positive integers.

\medskip
Given an instance of subset sum $I = (S,k)$, a crystal graph is created as follows.
For each integer $i \in S$ a new vertex with a charge of $i$ is created, note these correspond to the set $V^+$.
Two further ions are created, the first having a charge of $-k$ and the second having a charge of $-\left(\left(\sum\limits_{v_i \in V^+} wt(v_i)\right) - k\right)$, these correspond to $V^-$.
The value $k'$ is chosen as the greater of $k$ and $\left(\sum\limits_{v_i \in V^+} wt(v_i)\right) - k$.

Given $I$, we claim that the only minimal $k'$-Charge Removal, $R$, from $S = (V,E)$ is $R = V$ if and only if there is no solution to $I$.
To disprove that $R$ is minimal there must be some subset $R' \subset R$ that is also a set of $k_{\geq}$-charges.
As $k'$ charges must be removed, any such $V'^-$ must only contain the vertex in $V^-$ with a charge of $-k'$.
Therefore if this claim is false, there must be a set $R'^+ \subseteq R^+$ such that $\sum\limits_{v_i \in R'^+} wt(v_i) = k'$.
If there is such a $R'$ then there must also be a solution to the subset sum instance as either $R'^+$ or $R^+ \setminus R'^+$.
This is shown as if $k' = k$, the values in $R'^+$ must sum to $k$, satisfying $I$.
Conversely, if there is $k'$-Charge removal $R \subset V$ then following the above arguments, there must be a solution to $I$.

In the other direction, if there is a solution to $I$ then trivially there must be exist such a $R'^+$ that would make $R'$ non-minimal.
Similarly if there is no valid solution to $I$ then the only minimal set of $k$-charges is the complete set of ions.
Therefore it may not be determined if a solution is minimal in polynomial time.
Subsequently as a minimal set of charges for \akcr is required, a solution can not be verified in polynomial time unless $P = NP$, therefore it is not in NP in the general case.
\end{proof}

\begin{corollary}
It is NP-Hard to determine if an instance of \kcr has a valid solution in the case there are no bounds on the charges of the vertices.
\label{col:np_for_sol_kcr}
\end{corollary}

\begin{proof}
It follows from the arguments of Proposition \ref{prop:kcr_not_np} that an instance of \kcr may be constructed for a subset sum instance $I = (S,k)$ such that it is only satisfiable if the subset sum instance is.
\end{proof}

\begin{lemma}
A set of $k$-Charges may be verified as minimal in polynomial time for charges bounded by a polynomial size.
\label{lem:minkcr_in_np}
\end{lemma}

\begin{proof}
In the case when the charges of the vertices are bounded, a solution to the subset sum may be found in polynomial time, for example, relative to either the upper limit on the weights due to Pisinger \cite{pisinger1999linear}, or the number of distinct weights and the goal values due to Axiotis and Tzamos \cite{axiotis2018capacitated}.
Using these algorithms a set of $k$-Charges $R$ can be verified as minimal.
This is done by, for every value $1 \leq t \leq \sum\limits_{v_i \in R^+} wt(v_i) - k$ checking if there is a subset of charges $R'^+ \subseteq R^+$ and $R'^- \subseteq R^-$ such that $t =  \sum\limits_{v_i \in R'^+} wt(v_i)  = \left\lvert \sum\limits_{v_i \in R'^-} wt(v_i)\right\rvert$.
If there exists such a solution for any $t$ then $R$ is not minimal.

The claimed energy may also be verified by checking the sum of pairwise interactions relative to $U$, which may be trivially done in polynomial time due to the definition of $U$.
Therefore under these restrictions \akcr is in NP.
\end{proof}

\section{NP-Hardness for an unbounded number of ion species}
\label{sec:kcr}

\noindent
This section focuses on the class of potential functions \fcal.
It is assumed that the energy function for all cases is an arbitrary function in \fcal for which the parameters required by the ions to result in the energy from their pairwise interaction to be any arbitrary $a$ are known.
NP-completeness for \kcr as well as for the generalisations to \minkcr and \akcr is shown when there are bounds on value of the charges (either quantity of charges, or the maximum value).
Further, these problems are shown to be APX-Hard for bounded values of charges.
It may be noted that in the case the charges are not bounded, \minkcr remains NP-Hard, however as it is not in NP it is not NP-complete.
Along with the hardness results we provide a polynomial time reduction from both \kcr and \minkcr to max-weight-$k$-clique, under the restriction that all vertices have a charges of $\pm c$ for some non zero $c \in \mathbb{Z}$.

\begin{theorem}
\kcr, \minkcr and \akcr are NP-Complete for any energy function in $\mathcal{F}$ for vertices with charges of $\pm c$, for any natural number~$c$.
\label{thm:kcr-np}
\end{theorem}

\begin{proof}
\kcr and \akcr are in NP by Proposition \ref{prop:kcr_in_np}, and as the vertex charges are bounded, \minkcr is in NP by Lemma \ref{lem:minkcr_in_np}.
Hardness is established via a reduction from \textsc{clique}.
This is shown by reduction to \kcr, noting that any satisfying solution to \kcr also satisfies \akcr and \minkcr.

In the Clique problem, henceforth \textsc{clique}, the input is a graph, $G$, and a natural number, $k$.
The goal is to find a clique of size $k$ in $G$, or report that no such clique exists.
A clique is a set of vertices in a graph such that all vertices in the set are adjacent to each other.

\medskip
Given an instance of \textsc{clique}, $I = (G,k) = ((V,E),k)$, where $n = |V|$, an instance, $I'$, of \kcr is constructed as follows.
A unit cell of arbitrary size is chosen.
Within this cell $2n$ unique positions are created at arbitrary points in the unit cell.
In the first $n$ positive ions are placed and in the last $n$ negative ions are placed.
Each ion has its own unique specie.
Every vertex $v_i \in V$ corresponds to two ions, $i^+$ and $i^-$ with charges $c$ and $-c$ respectively.
For two ions $i$ and $j$ associated with $v_i$ and $v_j$ respectively the parameters are set so as to satisfy the following:
\[
    U_{ij} = \begin{cases}
        -1 & \text{$v_i = v_j$ or $(v_i, v_j) \in E$}\\
        p & \text{otherwise.}
    \end{cases}
\]
\noindent
Where $p \in \mathbb{R}$ is some arbitrarily high penalty value that may be treated as effectively being equal to $\infty$ for any practical purpose.
The definition of \fcal guarantees that there exists parameters satisfying these conditions irrespective of the positions, and thus the distance $r_{ij}$, of the ions.
Note that there are $k(2k-1)$ edges in a clique of size $2k$.
Let $g = k(2k-1)$ and let $k' = n - k$. To remove $k'$ positive and $k'$ negative ions $c k'$ vertices must be removed.

The corresponding crystal graph $G' = (V', E')$ is constructed as described in the preliminaries.
Let the vertices $v_i^+, v_i^- \in V'$ represent the ions corresponding to $v_i \in V$.
$v_i^\pm$ is used to denote either $v_i^+$ or $v_i^-$, where the charge of the vertex doesn't matter i.e. we are only concerned with the vertex in $G$ that $v_i^\pm$ corresponds to.
From the definition of the energy function, $wt(v_i^\pm, v_j^\pm) = -1$ if $i = j$ or $(v_i, v_j) \in E$, and $p$ otherwise.

We claim that $I$ is satisfiable if and only if $I'$ is satisfiable.
First consider the case that $I$ is satisfiable.
In this case $c k'$ vertices may be removed from $I'$, leaving only the vertices corresponding to the clique in $I$, denoted $A$.
As all vertices in $A$ correspond to adjacent vertices in $G$, the energy is $-1$ multiplied by the number of edges, giving a total energy of $-k(2k-1)$, satisfying the \kcr instance.
Conversely if there does not exist a clique of size $k$ in $G$ then any subset of charges $A \subseteq V'$ of cardinality $k$ clearly must contain at least one edge with a weight of $p$, making $I$ unsatisfiable.
\end{proof}

\noindent
This may be extended to other graph problems relatively easily.
One example of this would be the \textsc{max-weight k-clique} problem.
The \textsc{max-weight k-clique} problem takes as input a weighted graph $G$, a natural number $k$, and a goal value $v$.
The problem is to report if a clique of size $k$, where the sum of the weights of the edges is at least $v$ exists.
Using the above construction, a crystal graph $G'$ may be created from $G$.
From this the weights on the edges may be adjust as follows:
\[
    U_{ij} = \begin{cases}
        -wt(v_i, v_j) & \text{$v_i \neq v_j$ and $(v_i, v_j) \in E$}\\
        -c & v_i = v_j\\
       p & \text{otherwise.}
    \end{cases}
\]
Where $p \in \mathbb{R}$ is some arbitrary large penalty value, $wt(v_i,v_j)$ denotes the energy between vertices $v_i$ and $v_j$ in $G$ and $c$ is some constant such that $\nexists (v_i, v_j) \in E$ where $wt(v_i, v_j) \geq c$.
The goal value for the \kcr instance is chosen as $-k\cdot c - v$.
The correctness of this reduction follows from the arguments in Theorem \ref{thm:kcr-np}.

\begin{theorem}
\kcr remains NP-Hard for a given set of allowed charges with unique magnitude and an energy function within \fcal.
\label{thm:kcr_np_any_charges}
\end{theorem}

\begin{proof}
The construction of Theorem \ref{thm:kcr-np} may be extended to the case the set of vertices is limited to any set of allowed charges.
Two charges are chosen from this set, $c$ and $d$ where $|c| > |d|$ and $c  d < 0$ such that the difference between the absolute value of the charges, $|c| - |d|$, is minimised.
The same steps as in Theorem \ref{thm:kcr-np} are followed for the construction to get an initial crystal graph $G = (V, E)$ and $k'$.
Note that $I$ has a \emph{deficiency} of $n(|c| - |d|)$, where $n$ is the number of ions in the initial construction, meaning that some set of vertices must be added to make the cell neutral.
To handle the deficiency two sets of dummy vertices with charges of $c$ and $d$ are created.

The first set is to deal with the deficiency that would be left from a clique of size $k$.
To construct these, a natural number $t$ is chosen such that there exists a pair of natural numbers $t_c$ and $t_d$ such that $t_c  |c| = t$ and $t_d  |d| = k(|c| - |d|) + t$.
Using these, $t_c$ vertices with a charge of $c$ and $t_d$ vertices with a charge of $c$ and $t_d$ with a weight of $d$ are added.
From the definition of \fcal, the energy between them and all ions in $G$ and between each other is set as $0$.

The second set of dummy vertices are to counteract the overall deficiency in the initial unit cell.
A natural number $u$ is chosen such that there exists a pair of natural numbers $u_c$ and $u_d$, where $u = |c|  (u_c + t_c)$ and $u + n(|c| - |d|) = |d| (u_d + t_d)$.
$u_c$ vertices with a charge of $c$ and $u_d$ vertices with a charge of $d$ are added.
The potential energy between between them and all other vertices, including the set of previously added dummy vertices, is $P$.

To ensure that the optimal set of ions to be left with is a clique of size $k$ as well as all of the dummy vertices added in the first step, the following is done.
The goal energy remains the same as from Theorem \ref{thm:kcr-np}.
Observe that the only way to achieve this is to leave vertices corresponding to a clique of size at least $k$.
As there are $c (n + u_c + t_c)$ vertices for one set, and the goal is to be left with $c (k + t_c)$, a value $k'$ is chosen to remove as $c (n + u_c - k)$.
In the case that exactly $k'$-charges are removed, either the dummy vertices or some other vertices corresponding to a clique of size greater than $k$, ensuring the set remains neutral is left.
In the at-least-$k'$ case, some dummy vertices may also be removed provided the cell remains neutral.

From the arguments in Theorem \ref{thm:kcr-np} this is sufficient to ensure the new instance is satisfiable if and only if the original \textsc{clique} instance is.
Therefore these problems are NP-Hard, even in the case that there are distinct charges $c$ and $d$, $|c| \neq |d|$.
\end{proof}

\begin{theorem}
For any $\epsilon > 0$, \akcr for $k = 0$ cannot be approximated within a factor of $n^{1 - \epsilon}$, where $n$ is the number of ions, in polynomial time unless $P = NP$.
\label{thm:non_approximable}
\end{theorem}

\begin{proof}
This result follows from the results of H{\aa}stad \cite{Hastad1999}, who showed an approximation bound of $n^{1 - \epsilon}$ for Max-Clique.
Using the reductions from Theorems \ref{thm:kcr-np} and \ref{thm:kcr_np_any_charges} let the minimum energy after a removal of at least 0 vertices be $e$.
Note that this corresponds to the lowest potential energy of the instance.
From the reductions, it is clear that $e = -k$, where $k$ is the size of the clique, or $P$ if the remaining ions do not correspond to a clique.
Note that as $P$ may be arbitrarily large, given an approximation algorithm for any instance of \akcr with $k = 0$ an approximation algorithm for Max-Clique may be derived that approximates the instance of Max-Clique to the same factor as it approximates the \akcr instance.
Therefore any bounds on the approximation of Max-Clique must also apply to this problem, hence \akcr can not be approximated within a factor of $n^{1 - \epsilon}$ for any $\epsilon > 0$ within polynomial time unless $P = NP$.
\end{proof}

\noindent
While there are simple reductions to other NP-Complete problems such as Integer Programming, embedding this problem into many classical problems is made difficult due to the problem of maintaining the neutrality of the unit cell.
To this end, Theorem \ref{thm:kcr_max_weight_clique} provides a novel polynomial time reduction to show how a restricted version of \kcr may be embedded into \textsc{max-weight k-clique}.

\begin{theorem}
\label{thm:kcr_max_weight_clique}
\kcr can be reduced to \textsc{max-weight k-clique} in polynomial time, under the restriction that vertices are limited to charges of $\pm c$ and the energy function is computable within polynomial time.
\end{theorem}

\begin{proof}
Note that, given charges of $\pm c$, a valid solution to \minkcr is either valid for \kcr, or there is no valid solutions to \kcr.
Taking as input an instance of \minkcr with charges of $\pm c$ with the corresponding crystal graph, it is claimed that this instance may be represented as an instance of the weighted generalisation of \textsc{clique}.

In weighted $k$-clique, denoted \textsc{max-weight k-clique}, the input is a weighted graph, a goal value $v$, and a natural number $k$.
An instance of \textsc{max-weight k-clique} is satisfiable if and only if there exist a clique of size $k$ such the sums of the weight of the edges in the clique is at least $v$.

\medskip
Given an instance of \minkcr $I = (G,k) = \{(V, E), k\}$, an instance $I'$ of \textsc{max-weight k-clique} is created as follows.
A value $k'$ is chosen as $\frac{|V^+| - k}{c}$ rounded down to the nearest natural number.
The reason for this choice is to ensure that the optimal clique has size equal to the number of vertices left after removing $k$ charges.
Note that if $(|V^+| - k) (\mod c) \not\equiv 0$, then there is no valid solution to \kcr, however there may still be some valid solution to \minkcr.
A new graph $G' = (V', E')$ is created which is initially empty.
For each pair of vertices with different charges a new associated vertex in $V'$ is created.
An edge is created between each new vertex if and only if the corresponding charges are all unique, i.e. given the set of charges $V^+ = \{v_i, v_j\},$ and $v^- = \{v_k, v_l\}$ an edge would be placed between the new vertex representing $(v_i,v_k)$ and the one representing $(v_j, v_l)$, but not from either to the vertex representing $(v_i,v_l)$.
Give two connected vertices corresponding to charges $(v_i,v_k)$ and $(v_j,v_l)$ the edge between them is assigned a value of $- \left(U_{ij} + U_{il} + U_{jk} + U_{kl} + \frac{U_{ij} + U_{jl}}{k' - 1}\right)$.
The intuition behind this is for the edge to maintain the weights of the edges in $G$.
$\frac{U_{ik} + U_{jl}}{k' - 1}$ is added to this so that within a clique of size $k'$, the edge between the two vertices is fully represented.
An example of this construction is shown in Figure \ref{fig:kcr-maxClique}, omitting weights for legibility.

\begin{figure}[h]
\vspace*{4mm}
    \centering
    \includegraphics[scale=0.95]{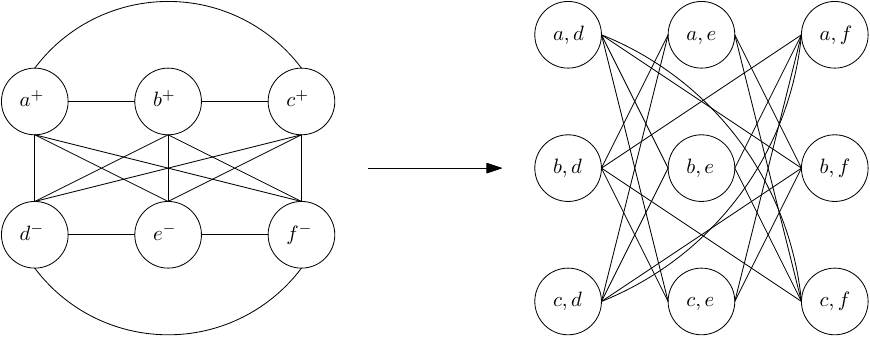}
    \caption{Example of the construction from \akcr to clique. Note that vertices $a,b,$ and $c$ have a positive weight, while $d,e,$ and $f$ have a negative weight. Also note that any clique of size 3 corresponds to the original graph.}
    \label{fig:kcr-maxClique}
\end{figure}

It is now be claimed that any clique of size $m$ corresponds to a neutrally weighted subset $A \subseteq V$, where $\left\lvert\sum\limits_{v_i \in A^+} wt(v_i)\right\rvert = m c$.
This is shown by noting that vertices are only connected if they do not represent a common vertex.
As such a clique of size $m$ must contain $m$ unique positively weighted and $m$ unique negatively charged vertices for the corresponding vertices to be connected as a clique.
Therefore by selecting any clique of size $k'$ in this graph, there is a valid structure left with exactly $k'$ unique positively weighted and $k'$ unique negatively charged vertices.
From the definition of $k'$ this corresponds to a subgraph of $G$ after a minimal removal of $k$.

\eject
It may now be claimed that a maximum weight clique of size $k'$ corresponds to the best subset of ions after a removal of $k$ charges.
Note that given a clique with total weight $w$ corresponds to a set of ions with total energy $-w$.
It is a straightforward extension to see that a maximum weight clique corresponds to a minimum energy subset of ions.
This is seen by noting that by choosing $k'$ as the size of the clique, the corresponding arrangement $A \subseteq V$ has $\left\lvert\sum\limits_{v_i \in A^+} wt(v_i)\right\rvert = c  k'$.
From the definition of $k'$, this requires $\left\lvert \sum\limits_{v_i \in A^+} wt(v_i)\right\rvert \leq \left\lvert\sum\limits_{v_i \in V^+} wt(v_i)\right\rvert - k$, which satisfies the requirements for a $k$-Charge Removal.
Conversely the definition of $k'$ ensures that the removal must be minimal.
Therefore the optimal solution to the \textsc{max-weight k-clique} instance must correspond to an optimal solution to the \minkcr \kcr instance.
Similarly any valid solution to the \textsc{max-weight k-clique} instance corresponds to some solution to the \minkcr instance.
\end{proof}

\section{Bounded number of species with Buckingham-Coulomb potential}
\label{sec:kcr_ind_set}

\noindent
In Section \ref{sec:kcr} NP-Hardness was shown for the case that there was an unbounded number of species, and NP-completeness in the case that there is a bounded number of charge values.
This is strengthened by considering instances with only two unique species.
Only the Buckingham-Coulomb potential function with charges of $\pm 1$ is considered in this section.
All three problems are again considered, noting that for charges of $\pm 1$ \kcr is equivalent to \minkcr.
NP-Hardness is shown by a reduction from Independent Set problem denoted \textsc{independent-set}, on penny graphs - adapting it to the Euclidean settings of a crystal graph of ions within a unit cell.
The Independent Set problem takes as input a graph, $G$, and a natural number $k$.
The goal is to find an \emph{independent set}, i.e. a set of vertices such that no two are adjacent, of size $k$ in $G$, or report that one does not exist.
A \emph{Penny graph} is a graph where each vertex may be drawn as a unit circle such that no two circles overlap, and an edge between two vertices exists if and only if the corresponding circles are tangent, i.e. they intersect at only a single point.
Finding an independent set on this class of graphs was shown to be NP-Hard by Cerioli et al. \cite{ITA_2011__45_3_331_0}.
The NP-Hardness result for this problem is shown by a reduction from \textsc{max-degree 3 planar vertex cover}, shown to be NP-Complete by Garey and Johnson \cite{doi:10.1137/0132071}.

\medskip \noindent
{\bf Construction of the \kcr instance:}
Let $I = (G,k)$ be an instance of \textsc{independent-set} where $G = (V,E)$ is a planar graph with a maximum degree of $3$ and $k \in \mathbb{N}$ is the size of the target independent set.
An instance of \kcr is created as follows.
Using Theorem 1.2 from Cerioli et al. create a new penny graph realisation, $G'$, and a new natural number $k$.
The class of graphs created by this process is denoted as the \emph{long orthogonal penny graphs}.
The radius of each circle for $G'$ is chosen as $\frac{n}{2}$.

A region of space in $\mathbb{R}^3$ with a height of at least $1$ and a width and length allowing $G'$ to be drawn is created.
This space is the parallelepiped for the unit cell.
In this space, two copies of $G'$ are drawn such that one is directly above the other at a distance of $1$.
For every circle in $G'$ two ions are created, one in the lower copy of $G'$ and the other in the higher copy.
Each ion is labelled with the vertex from $G'$ it corresponds to.
In this context \emph{pair} refers to the two ions in the new crystal graph, labelled with the same vertex from $G'$.
Two pairs are \emph{neighbouring} if they represent vertices that are adjacent in $G'$.
The lower ions are assigned the positive specie and the upper ions the negative.
An example of this arrangement is provided in Figure \ref{fig:ind_set_ex}.
Note that the minimum distance between two pairs in the same plane that are non-adjacent for circles with a radius of $\frac{n}{2}$ is $\sqrt{2}n$, as shown in Figure \ref{fig:penny_distances}.

The positive and negative species are assigned charges of $+1$ and $-1$ respectively.
In general there are 3 sets of parameters to choose determining the interaction between two positive ions, two negative ions, and one positive and one negative ion.
For simplicity, the parameters determining the interaction between two positive ions and the parameters determining the interaction between two negative ions are treated as being the same.
Informally, this means the the energy between two negative ions at a distance of $r$ from each other is the same as the energy between two positive ions at a distance of $r$.
For brevity, 1 and 2 are used to denote the positive and negative specie respectively.
With this notation, the parameters that may be set are $A_{11}, B_{11}, C_{11}, A_{12}, B_{12},$ and $C_{12}$.

\medskip
An independent set is said to be \emph{left} if the ions left after a removal of $k'$ charges have labels corresponding to an independent set in $G'$.
Let $k' = n - k$, be the number of charges that are required to be removed to be left with an independent set of size $k$.
Note that as the charge of each ion has a magnitude of one, a removal of $k'$ can only be achieved by removing $k'$ positive and $k'$ negative ions.
The goal energy for the construction is set as $g = (k - 1)(\frac{A_{12}}{e^{B_{12}}} - C_{12} - 1)$.
To simplify the equations regarding the interaction between planes, $\widehat{r}$ is used to denote $\sqrt{r^2 + 1}$.
To ensure that an independent set is left of size $k$ if and only if one exists, the following three inequalities must be satisfied:
\begin{align}
    \frac{A_{11}}{e^{B_{11} n}} - \frac{C_{11}}{n^6} + \frac{1}{n} + \frac{A_{12}}{e^{B_{12} \widehat{n}}} - \frac{C_{12}}{\widehat{n}^6} - \frac{1}{\widehat{n}} &\geq \left|\frac{A_{12}}{e^{B_{12}}} - C_{12} - 1\right|\label{ineq:pos_in_cirlce}\\
    n^2 \left|\frac{A_{11}}{e^{B_{11} r}} - \frac{C_{11}}{r^6} + \frac{1}{r} + \frac{A_{12}}{e^{B_{12} \widehat{r}}} - \frac{C_{12}}{\widehat{r}^6} - \frac{1}{\widehat{r}}\right| &\leq \left|\frac{A_{12}}{e^{B_{12}}} - C_{12} - 1\right|, &r \geq \sqrt{2}n\label{ineq:neg_out_cirlce}\\
    \frac{A_{11}}{e^{B_{11} r}} - \frac{C_{11}}{r^6} + \frac{1}{r} + \frac{A_{12}}{e^{B_{12} \widehat{r}}} - \frac{C_{12}}{\widehat{r}^6} - \frac{1}{\widehat{r}} &> 0, &r \geq \sqrt{2}n \label{ineq:same_plane_pos}
\end{align}

\noindent
At a high level, these inequalities are used as follows.
Inequality \ref{ineq:pos_in_cirlce} ensures that the positive interaction between some ion $i$ and the adjacent pair is greater than the interaction between $i$ and the other ion $j$ in the same pair.
This ensures that the cost of keeping two adjacent pairs is greater than the negative energy gained by keeping it.
Inequality \ref{ineq:neg_out_cirlce} ensures that the total positive interaction energy between some ion $i$ and every ion further than $n$ is no more than the negative energy given by maintaining the pair containing $i$.
This ensures that, given a pair of ions representing a vertex that is not adjacent to any other pairs, it is better to keep the pair in the structure rather than to remove the pair.
Finally, Inequality \ref{ineq:same_plane_pos} is used to ensure that the interaction between two ions on the same plane leads to a positive energy penalty no matter the distance.
This is used to bound the total energy from above.

\medskip
The following Lemmas use these inequalities as follows.
Lemma \ref{lem:ineq_satisfiable} show that there exists some parameters for the Buckingham-Coulomb potential satisfying the inequalities, allowing them to be used as a tool when considering the subsequent Lemmas.
Lemmas \ref{lem:ind_set_ineq} and \ref{lem:ind_set_energy_bounds} assume that it is preferable to choose $k'$ pairs over any other set of charges in the arrangement.
Lemma \ref{lem:ind_set_ineq} shows that the optimal removal will result in an independent set.
Lemma \ref{lem:ind_set_energy_bounds} provides an upper bound on the interaction energy for this setting.
Finally Lemma \ref{lem:pairs_pref} proves that it is always preferable to choose $k'$ pairs over any other set of charges in the arrangement.

\begin{lemma}
\label{lem:ineq_satisfiable}
There exists, for any structure created from a long orthogonal penny graph, some parameters for the Buckingham-Coulomb potential such that Inequalities (\ref{ineq:pos_in_cirlce}, \ref{ineq:neg_out_cirlce}) and (\ref{ineq:same_plane_pos}) are satisfied.
\end{lemma}

\begin{proof}
Values are chosen for $A_{12}, B_{12}$ and $C_{12}$ such that the energy for any pair of ions of opposite vertex at a distance of 1 is $-1$.
This is achieved by choosing a value of $\frac{1}{2n^2}$ for $A_{12}$, 0 for $B_{12}$, and $\frac{1}{2n^2}$ for $C_{12}$.
This simplifies the energy equation to:
\begin{equation*}
    U^{BC}(r) = \frac{A_{11}}{e^{B_{11}r}} - \frac{C_{11}}{r^6} + \frac{1}{r} + \frac{1}{2n^2} - \frac{1}{2 n^2 \widehat{r}^6} - \frac{1}{\widehat{r}}.
\end{equation*}

\noindent
To satisfy Inequality \ref{ineq:pos_in_cirlce}, $U^{BC}(n) > 1$.
This may be satisfied by choosing values for $A_{11}$, $B_{11}$, and $C_{11}$ such that $\frac{A_{11}}{e^{B_{11}n}} - \frac{C_{11}}{n^6} = 1$, noting that $\frac{1}{2n^2} - \frac{1}{2 n^2 \widehat{r}^6} + \frac{1}{r} - \frac{1}{\widehat{r}} > 0$ for all positive distances greater than~1.
This is satisfied by solving the equation $\frac{A_{11}}{e^{B_{11}n}} - \frac{C_{11}}{n^6} = 1$, choosing $A_{11} = \frac{C_{11} e^{B_{11}n}}{n^6} + e^{B_{11}n}$.

\medskip
Inequality (\ref{ineq:neg_out_cirlce}) requires that at a distance of at least $\sqrt{2}n$ the total energy is no more than $\frac{1}{n^2}$.
This is satisfied at a distance of $\sqrt{2}n$ by ensuring that the $\frac{A_{11}}{e^{B_{11}\sqrt{2}n}} - \frac{C_{11}}{8n^6} = 0$, which is satisfied after substituting in the appropriate value for $A_{11}$ with $C_{11} = \frac{e^{B_{11}n}}{e^{B_{11}\sqrt{2}n}\left(\frac{1}{\left(\sqrt{2}n\right)^6} - \frac{e^{B_{11}n}}{n^6e^{B_{11}\sqrt{2}n}}\right)}$, which simplifies to $\frac{8n^6}{e^{(\sqrt{2} - 1)B_{11} \cdot n} - 8}$.
Finally, consider the value of $B_{11}$.
Note that the value of both $A_{11}$ and $C_{11}$ depend greatly on $B_{11}$, with a small increase in $B_{11}$ leading to a very rapid increase in the value of $A_{11}$ and a rapid decrease in $C_{11}$.
Similarly, the value of the energy given by $\frac{A_{11}}{e^{B_{11}r}} - \frac{C_{11}}{r^6}$ rapidly decreases initially before converging at approximately 0, noting that the first derivative with respect to $r$ of this equation for a given $n$ is $-B_{11}\frac{A_{11}}{e^{B_{11}r}} + \frac{6C_{11}}{r^7}$.
As such by choosing a suitably large $B_{11}$ Inequality (\ref{ineq:neg_out_cirlce}) is easily satisfied, one obvious choice for this would be $B_{11} = n$.

\medskip
Note that both $\frac{A_{11}}{e^{B_{11}r}}$ and $\frac{C_{11}}{r^6}$ strictly decrease, therefore if $\left|\frac{A_{11}}{e^{B_{11}r}}\right| \leq \frac{1}{2n^2} - \frac{1}{r} + \frac{1}{\widehat{r}}$ and $\left|\frac{C_{11}}{r^6}\right| \leq \frac{1}{r} - \frac{1}{\widehat{r}}$ both Inequalities (\ref{ineq:neg_out_cirlce}) and (\ref{ineq:same_plane_pos}) are satisfied.
From the value of $B_{11}$ this becomes $\frac{n^6}{e^{(\sqrt{2} - 1)n^2} - 8}$, which is positive and less than $\frac{1}{n^2}$ for any $n \geq 6$.
Note that $\frac{1}{r} - \frac{1}{\widehat{r}} \geq \frac{1}{r^6}$ for any $r \geq 1.3$, hence it is clear that $\frac{C_{11}}{r^6} \leq \frac{1}{r} - \frac{1}{\widehat{r}}$ for $n \geq 6$.
Considering $\frac{A_{11}}{e^{B_{11}r}}$ at a distance of $\sqrt{2}n$, the equation becomes $\frac{A_{11}}{e^{n^2 \sqrt{2}}} = \frac{C_{11}}{n^6 e^{n^2 \sqrt{2}}} + \frac{1}{e^{(\sqrt{2} - 1)n^2}}$, from the previous arguments it follows that this is considerably less than $\leq \frac{1}{2n^2} - \frac{1}{r} + \frac{1}{\widehat{r}}$ for $n \geq 6$.

Noted that due to the constant $\frac{1}{2n^2}$ term there is a positive value for any distance greater than $\sqrt{2}n$, satisfying Inequality (\ref{ineq:same_plane_pos}).
Using these values, it has now been shown how to design parameters for the Buckingham-Coulomb potential which satisfy the Inequalities.
\end{proof}

\begin{lemma}
\label{lem:ind_set_ineq}
Inequalities (\ref{ineq:pos_in_cirlce}) and (\ref{ineq:neg_out_cirlce}) are sufficient to ensure that an independent set is left if one exists in the original \textsc{independent-set} instance.
\end{lemma}

\begin{proof}
Inequality (\ref{ineq:pos_in_cirlce}) ensures that if there are two pairs corresponding to points that intersect, the total energy always decreases by removing one of the pairs.
Inequality (\ref{ineq:neg_out_cirlce}) complements (\ref{ineq:pos_in_cirlce}) by ensuring that given a pair corresponding to a vertex with no adjacent neighbours, the total energy would increase by removing it.
This holds even in the case that all other pairs are at a distance of $\sqrt{2}n$.
Inequalities (\ref{ineq:pos_in_cirlce}) and (\ref{ineq:neg_out_cirlce}) combined means that the global minimum total energy for any subset is the maximum independent set.
Note that the total energy decreases with the cardinality of the given independent set.
\end{proof}

\begin{lemma}
\label{lem:ind_set_energy_bounds}
Given $k$ pairs, the energy is less than $(k-1)(\frac{A_{12}}{e^{B_{12}}} - C_{12} - 1)$ if and only if the pairs correspond to an independent set of size $k$, for $\frac{A_{12}}{e^{B_{12}}} - C_{12} - 1 < 0$.
\end{lemma}

\begin{proof}
Given $k$ pairs, the energy between the ions in each pair is $\frac{A_{12}}{e^{B_{12}}} - C_{12} - 1$, for a total of $k(\frac{A_{12}}{e^{B_{12}}} - C_{12} - 1)$.
Inequality (\ref{ineq:neg_out_cirlce}) ensures that the maximum energy gained from pairs of ions corresponding to non-intersecting circles is at most $|\frac{A_{12}}{e^{B_{12}}} - C_{12} -1|$. Inequality (\ref{ineq:same_plane_pos}) ensures that having vertices on the same plane leads to a slight positive charge.
From this it follows that the maximum energy a set of ions corresponding to an independent set is $(k-1)(\frac{A_{12}}{e^{B_{12}}} - C_{12} - 1)$.
Conversely, from Inequality (\ref{ineq:pos_in_cirlce}) it is known that if there is a pair of intersecting circles the total energy must be greater than  $(k-1)(\frac{A_{12}}{e^{B_{12}}} - C_{12} - 1)$.

Note that for \akcr that if greater than $k'$ pairs were removed this energy could not be achieved as the minimum energy would be $(k-1)(\frac{A_{12}}{e^{B_{12}}} - C_{12} - 1)$ for the interaction within pairs.
As there is a positive interaction between pairs, the total energy must be slightly greater than this for any $k > 1$.
Therefore the total energy is less than $(k-1)(\frac{A_{12}}{e^{B_{12}}} - C_{12} - 1)$ if and only if there is an independent set of size $k$ left.
Note that under the choice of variables from Lemma \ref{lem:ineq_satisfiable}, the upper bound is $-(k-1)$.
\end{proof}

\begin{lemma}
When removing $k'$ charges from the construction from a long orthogonal penny graph, it is always preferable to remove pairs provided that Inequalities (\ref{ineq:pos_in_cirlce}-\ref{ineq:same_plane_pos}) hold.
\label{lem:pairs_pref}
\end{lemma}

\begin{proof}
Assume that this statement is false, there must be some assignment, where it is preferable to remove some set of at least two vertices, $v_i^+$ and $v_j^-$, that do not form a pair with any ions that have been removed.
Assume that there are $t$ positive and $t$ negative vertices in the graph.
If instead $v_j^-$ was left in, while $v_i^+$ was removed, the remaining energy would change by at least $-1 + t\left(\frac{A_{11}}{e^{B_{11} \sqrt{2}n}} - \frac{C_{11}}{(\sqrt{2}n)^6} + \frac{1}{\sqrt{2}n} + \frac{A_{12}}{e^{B_{12} \widehat{\sqrt{2} n}}} - \frac{C_{12}}{\widehat{\sqrt{2} n}^6} - \frac{1}{\widehat{\sqrt{2}n}}\right)$.
From the arguments in Lemma \ref{lem:ind_set_ineq} and the construction in Lemma \ref{lem:ineq_satisfiable} this leads to a decrease in total energy, making it preferable and therefore contradicting the assumption.
Note that given a positively charged vertex of the maximum degree, in this case 3, it could contribute at most $\frac{3}{2n^2} - \frac{3}{2n^8} - \frac{3}{n}$ which has a magnitude less than 1 for any $n \geq 3$.
Therefore, by contradiction this holds.
\end{proof}

\begin{theorem}
\kcr, \minkcr and \akcr are NP-Complete when limited to only two species of ion and restricted to the Buckingham-Coulomb potential energy function.
\label{thm:kcr_2_np}
\end{theorem}

\begin{proof}
Building on the results from Lemmas \ref{lem:ind_set_ineq}, \ref{lem:ineq_satisfiable}, \ref{lem:ind_set_energy_bounds}, and \ref{lem:pairs_pref}, the next step is to show NP-Completeness.
Lemma \ref{lem:ind_set_ineq} shows that, provided Inequalities  (\ref{ineq:pos_in_cirlce}) and (\ref{ineq:neg_out_cirlce}) hold, the optimal solution is to leave an independent set.
From Lemma \ref{lem:ineq_satisfiable} it follows that these inequalities are satisfiable for any graph under the given construction, noting that the assignment of parameters gives an energy of $-1$ within pairs.
Lemma \ref{lem:ind_set_energy_bounds} shows that the upper bound is reachable if and only if an independent set has been left.
It follows from Lemma \ref{lem:pairs_pref} that the assumption that it is preferable to remove a set of pairs over any other set of charges holds when the inequalities also do.

Therefore there is a satisfiable instance of \kcr or any generalisation if and only if the instance of \textsc{independent set} for the maximum degree 3 planar graph instance is satisfiable.
Conversely if the \textsc{independent set} instance is satisfiable, the corresponding \kcr instance is satisfied by leaving the vertices corresponding to the independent set in the long orthogonal penny graph construction.
Hence under these restriction all three problems are NP-complete.
Note that this may be extended to vertices with charges $\pm c$ for any given $c$.
\end{proof}

\begin{corollary}
It is NP-Hard to approximate the optimal solution to \kcr, \minkcr, or \akcr within a factor $1 + \frac{3}{n - k - 1}$ for the Buckingham-Coulomb potential.
\end{corollary}\vspace*{-1mm}

\begin{figure}[!b]
\vspace*{-2mm}
    \centering
    \includegraphics[scale=0.92]{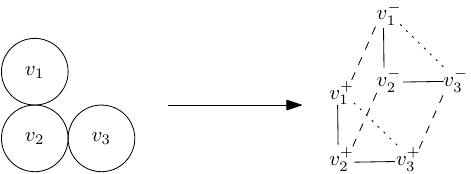}\vspace*{-2mm}
    \caption{Example of the construction of an arrangement from a penny graph, in this example $v_1$ and $v_2$ are adjacent, as are $v_2$ and $v_3$, but not $v_1$ and $v_3$} \label{fig:ind_set_ex}\vspace*{-1mm}
\end{figure}

\begin{proof}
Following Lemma \ref{lem:ind_set_energy_bounds}, given a set of $k$ ion pairs corresponding to an independent set the total energy is $(k-1)(\frac{A_{12}}{e^{B_{12}}} - C_{12} - 1)$.
In terms of the \kcr problem this can be rewritten as $(n - k - 1)(\frac{A_{12}}{e^{B_{12}}} - C_{12} - 1)$.
In the case that there exists some pair of ion-pairs representing adjacent vertices, there must be an energy penalty of $4\left(\frac{A_{11}}{e^{B_{11} n}} - \frac{C_{11}}{n^6} + \frac{1}{n} + \frac{A_{12}}{e^{B_{12} \widehat{n}}} - \frac{C_{12}}{\widehat{n}^6} - \frac{1}{\widehat{n}}\right)$.
Therefore, a lower bound on the energy in the case that there is at least some pair of ions adjacent to each other is $(n - k)(\frac{A_{12}}{e^{B_{12}}} - C_{12} - 1) + 4\left(\frac{A_{11}}{e^{B_{11} n}} - \frac{C_{11}}{n^6} + \frac{1}{n} + \frac{A_{12}}{e^{B_{12} \widehat{n}}} - \frac{C_{12}}{\widehat{n}^6} - \frac{1}{\widehat{n}}\right) \geq (n - k - 1)(\frac{A_{12}}{e^{B_{12}}} - C_{12} - 1) + 4\left(\frac{A_{11}}{e^{B_{11} n}} - \frac{C_{11}}{n^6} + \frac{1}{n} + \frac{A_{12}}{e^{B_{12} \widehat{n}}} - \frac{C_{12}}{\widehat{n}^6} - \frac{1}{\widehat{n}}\right)$.
Therefore any approximation algorithm that can achieve an approximation ratio smaller than $\frac{(n - k - 1)(\frac{A_{12}}{e^{B_{12}}} - C_{12} - 1) + 3\left(\frac{A_{11}}{e^{B_{11} n}} - \frac{C_{11}}{n^6} + \frac{1}{n} + \frac{A_{12}}{e^{B_{12} \widehat{n}}} - \frac{C_{12}}{\widehat{n}^6} - \frac{1}{\widehat{n}}\right)}{(n - k - 1)(\frac{A_{12}}{e^{B_{12}}} - C_{12} - 1)} = 1 + \frac{3\left(\frac{A_{11}}{e^{B_{11} n}} - \frac{C_{11}}{n^6} + \frac{1}{n} + \frac{A_{12}}{e^{B_{12} \widehat{n}}} - \frac{C_{12}}{\widehat{n}^6} - \frac{1}{\widehat{n}}\right)}{(n - k - 1)(\frac{A_{12}}{e^{B_{12}}} - C_{12} - 1)} \geq 1 + \frac{3}{n - k - 1}$ would be able to find the optimal solution to the underlying $k$-independent set problem in polynomial time unless $P = NP$.
\end{proof}

\begin{figure}[!h]
    \centering
    \includegraphics[scale=1]{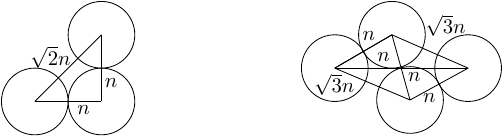}
    \caption{Illustration of the distances between the centre of non-adjacent pennies using the construction of Cerioli et al. \cite{ITA_2011__45_3_331_0}.}
    \label{fig:penny_distances} \vspace*{-3mm}
\end{figure}

\section{Restriction to the Coulomb potential with unbounded charges}
\label{sec:kcr_kanpsack}

\noindent
The final case that is considered in this work is when the energy function is the Coulomb potential.
NP-Hardness for this case is shown by a reduction from \textsc{knapsack} to \akcr.
Note that with an unbounded number charge values this problem is not in NP for \minkcr due to Proposition \ref{prop:kcr_not_np} and is trivially NP-Hard for \kcr due to Corollary \ref{col:np_for_sol_kcr}.
This reduction requires using an unbounded number of charge values, thus it follows from proposition \ref{prop:kcr_not_np} that it is NP-Hard to verify if a solution to an instance of \akcr is minimal.
In this reduction it is shown that an instance of \akcr such that the set of ions left correspond to the items for the knapsack instance if and only if there is a set satisfying the knapsack instance.

\begin{theorem}
\label{thm:kcr_coulomb_knapsack}
\akcr and \minkcr remains NP-Hard when the energy function is limited to the Coulomb potential.
\end{theorem}

\begin{proof}
In the knapsack problem, henceforth \textsc{knapsack}, the input is a bag with capacity $C$, and a set of items $S$.
Each item $i \in S$ has a weight $w_i$, and a value $p_i$.
In this problem the goal is to find the subset $S' \subseteq S$ such that $\sum\limits_{i \in S'} p_i$ is maximised conditional on $\sum\limits_{i \in S'} w_i \leq C$.
Alternatively this may phrased as a decision problem by taking some goal value $g$ and asking if there is an $S'$ such that $\sum\limits_{i \in S'} p_i \geq g$.

\medskip
NP-completeness for \akcr and \minkcr is shown by a reduction from the knapsack problem.
Given an instance, $I$, of the knapsack problem as described above, an instance, $I'$, of \akcr is created as follows.
For every $i \in S$, two charges are created denoted $v_i^+$ and $v_i^-$ and label with the corresponding item.
These are assigned a weight of $w_i$ to $v_i^+$ and $-w_i$ to $v_i^-$.

\medskip
The values $u$ and $\alpha$ are defined such that $u$ is some value such that there does not exist any pair of items, $i$ and $j$, such that $p_i > p_j$ but $p_j + u \geq p_i$, and $u$ is less than the smallest \emph{unit of precision} for the value of the items.
Using this, $\alpha$ is defined as some value satisfying the inequality $u > \frac{4n^2w^2_{max}}{\alpha}$, where $w_{max}$ is the weight of the heaviest item.
This ensures that $\alpha$ is some distance such that if all vertices are at least $\alpha$ away from each other there is a difference of no more than $u$ in energy, which is sufficient to ensure that vertices at that distance may be safely ignored.

These vertices are now placed such that, for each item, the distance between the two vertices $v_i^+$ and $v_i^-$ has a potential energy of $-p_i$.
Recall that $U^C_{ij} = \frac{q_i q_j}{r_{ij}}$. This is achieved by placing them at a distance of $\frac{w^2_{i}}{p_i}$.
Each of the pair of vertices representing an item is placed in a line so that the distance between any two pairs is no less than $\alpha$.
An example of this construction is provided in Figure \ref{fig:knapsack_example}.

\begin{figure}[!h]
    \centering
    \begin{tabular}{l l l}
        \emph{Item} & \emph{Weight} & \emph{Value}\\
        $I_1$ & 9 & 3 \\
        $I_2$ & 6 & 2 \\
        $I_3$ & 3 & 3
    \end{tabular}
    \includegraphics[scale=0.74]{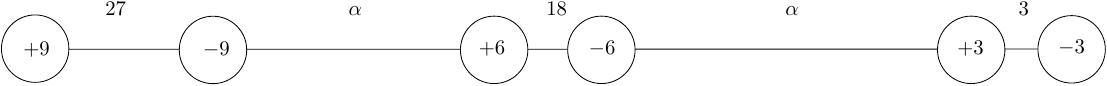}
    \caption{Example of construction of the structure from the knapsack instance. In this $u < 1$ and $\alpha > 2916$}
    \label{fig:knapsack_example}
\end{figure}

The value $k$ is chosen as $k = \left(\sum\limits_{i \in S} w_i \right) - c$, ensuring that there are no more than $c$ vertices left after removing $k$, corresponding to a valid assignment for the knapsack instance.
Finally, the goal value is chosen as $g' = -g + u$.

It follows from this construction that any removal of $k_{\geq}$ charges are a valid packing in terms of the capacity.

\medskip If the \akcr instance is satisfiable then there must be some valid packing of no more than $g'$ energy.
As the interaction between vertices corresponding to different items is trivially small, the only way to achieve this is to choose a set of vertex pairs with an energy between them no more than $g'$.
As the energy between pairs is equal to the value of the items, the only way this is achieved is to have items corresponding to a packing with value at least $g$.
Conversely if the \akcr instance is not satisfiable, there does not exist a packing of value $g$ by the same arguments.

\medskip
Similarly if the \textsc{knapsack} instance is satisfiable then the \akcr instance may be satisfied by removing all charges not corresponding to a satisfying packing of the \textsc{knapsack} instance.
Finally if the \textsc{knapsack} instance is not satisfiable then by the previous arguments the \textsc{kcr} instance also can not be satisfied.
Therefore this problem is NP-Complete.
Note that as the weights on all items is positive, with a corresponding negative energy in $I'$, given a non-minimal satisfying solution there exists some minimal satisfying solution.
Therefore \minkcr is also NP-Hard.
\end{proof}


\section{Conclusions and future work}
In this work we have presented the new problem of \kcr, and a class of functions for which the general case is NP-Complete.
In the general case, we show that this problem is APX-Hard.
We have also shown that the problem remains NP-Complete under both the restriction that we have only two species of ions and the Buckingham-Coulomb energy function and the restriction we only use the Coulomb potential on an unbounded number of ion species.

These results provide a strong indication on the general hardness of \textsc{csp}.
By constructing an instance of \kcr, where every position is simultaneously occupied by every possible ion, the optimal solution to \kcr for some $k$ is also the optimal solution to \textsc{csp}.
While this paper does not look at the this specific case, the hardness of \kcr both in general, and under the restriction to 2 ion species, suggests that \textsc{csp} is also NP-Hard.
These results do show that \textsc{csp} when restricted to the constructions used in the paper is NP-Hard.
More specifically, \textsc{csp} is NP-hard for a discrete grid, where each position has some restriction on the allowed ion species.

From a chemistry stand point, while we have made progress towards physical constructions there is still a lot that could be done in this regard.
As such investigation into the restrictions of having more realistic physical values remains an important unexplored direction.
Another question would be if we can investigate the convergence of these interactions, particularly the Coulomb potential, over a periodic structure to more fully understand the energy function.

\subsection*{Acknowledgements}
\noindent
The authors thank the Leverhulme Trust for funding this research via the Leverhulme Research Centre for Functional Materials Design.
This work was partialy supported by EPSRC grants EP/R018472/1 ``Application driven Topological Data Analysis'' and  EP/R005613/1 ``ALGOUK - A Network for Algorithms and Complexity in the UK''.  The research of Prof. Potapov has been partially supported by Leverhulme Trust Senior Research Fellowship (SRF\textbackslash R1\textbackslash 201074).


\begin{thebibliography}{10}
\providecommand{\url}[1]{\texttt{#1}}
\providecommand{\urlprefix}{URL }
\expandafter\ifx\csname urlstyle\endcsname\relax
  \providecommand{\doi}[1]{doi:\discretionary{}{}{}#1}\else
  \providecommand{\doi}{doi:\discretionary{}{}{}\begingroup
  \urlstyle{rm}\Url}\fi
\providecommand{\eprint}[2][]{\url{#2}}

\bibitem{Adamson2020}
Adamson D, Deligkas A, Gusev VV, Potapov I.
\newblock {On the Hardness of Energy Minimisation for Crystal Structure  Prediction}.
\newblock In: Lecture Notes in Computer Science), volume 12011 LNCS. Springer,
  2020 pp. 587--596.  	arXiv: 1910.12026 [cs.CC].

\bibitem{LYAKHOV20101623}
Lyakhov AO, Oganov AR, Valle M.
\newblock How to predict very large and complex crystal structures.
\newblock \emph{Computer Physics Communications}, 2010.
\newblock \textbf{181}(9):1623--1632.  doi:10.1016/j.cpc.2010.06.007.

\bibitem{woodley2008crystal}
Woodley SM, Catlow R.
\newblock Crystal structure prediction from first principles.
\newblock \emph{Nature materials}, 2008.
\newblock \textbf{7}(12):937--946.   doi:10.1038/nmat2321.

\bibitem{csp-local}
Antypov D, Deligkas A, Gusev V, Rosseinsky MJ, Spirakis PG, Theofilatos M.
\newblock Crystal Structure Prediction via Oblivious Local Search.
\newblock \emph{CoRR}, 2020.
\newblock \textbf{abs/2003.12442}.

\bibitem{doi:10.1002/1521-3765(20020916)8:18<4102::AID-CHEM4102>3.0.CO;2-3}
Mellot-Draznieks C, Girard S, Férey G, Schön JC, Cancarevic Z, Jansen M.
\newblock Computational Design and Prediction of Interesting
  Not-Yet-Synthesized Structures of Inorganic Materials by Using Building Unit
  Concepts.
\newblock \emph{Chemistry – A European Journal}, 2002.
\newblock \textbf{8}(18):4102--4113. doi:10.1002/1521-3765(20020916) 8:18<4102.

\bibitem{oganov2006crystal}
Oganov AR, Glass CW.
\newblock Crystal structure prediction using ab initio evolutionary techniques:
  Principles and applications.
\newblock \emph{The Journal of chemical physics}, 2006.
\newblock \textbf{124}(24). doi:10.1063/1.2210932.

\bibitem{WANG20122063}
Wang Y, Lv J, Zhu L, Lu S, Yin K, Li Q, Wang H, Zhang L, ma Y.
\newblock Materials discovery via CALYPSO methodology.
\newblock \emph{Journal of physics. Condensed matter : an Institute of Physics
  journal}, 2015.
\newblock \textbf{27}:203203.  doi:10.1088/0953-8984/27/20/203203.

\bibitem{Oganov2018}
Oganov AR.
\newblock Crystal structure prediction: reflections on present status and  challenges.
\newblock \emph{Faraday Discuss.}, 2018.
\newblock \textbf{211}:643--660. doi:10.1039/C8FD90033G.

\bibitem{Barahona_1982}
Barahona F.
\newblock On the computational complexity of Ising spin glass models.
\newblock \emph{Journal of Physics A: Mathematical and General}, 1982.
\newblock \textbf{15}(10):3241--3253.  doi:10.1088/0305-4470/15/10/028.

\bibitem{Wille_1985}
Wille LT, Vennik J.
\newblock Computational complexity of the ground-state determination of atomic  clusters.
\newblock \emph{Journal of Physics A: Mathematical and General}, 1985.
\newblock \textbf{18}(8):L419--L422.  doi:10.1088/0305-4470/18/8/003.

\bibitem{doi:10.1137/0208049}
Krishnamoorthy M, Deo N.
\newblock Node-Deletion NP-Complete Problems.
\newblock \emph{SIAM Journal on Computing}, 1979.
\newblock \textbf{8}(4):619--625.

\bibitem{doi:10.1137/0210022}
Yannakakis M.
\newblock Node-Deletion Problems on Bipartite Graphs.
\newblock \emph{SIAM Journal on Computing}, 1981.
\newblock \textbf{10}(2):310--327.  doi:10.1137/0210022.

\bibitem{Yannakakis78node-and}
Yannakakis M.
\newblock Node-and edge-deletion NP-complete problems.
\newblock In: Conference record of the tenth annual ACM Symposium on Theory or
  Computing (STOC), ACM. 1978 pp. 253--264.

\bibitem{Faber2015}
Faber F, Lindmaa A, von Lilienfeld OA, Armiento R.
\newblock {Crystal structure representations for machine learning models of
  formation energies}.
\newblock \emph{International Journal of Quantum Chemistry}, 2015.
\newblock \textbf{115}(16):1094--1101.
\newblock \doi{10.1002/qua.24917}.

\bibitem{collins2017accelerated}
Collins C, Dyer M, Pitcher M, Whitehead G, Zanella M, Mandal P, Claridge J,
  Darling G, Rosseinsky M.
\newblock Accelerated discovery of two crystal structure types in a complex
  inorganic phase field.
\newblock \emph{Nature}, 2017.
\newblock \textbf{546}(7657):280. doi:10.1038/nature22374.

\bibitem{ITA_2011__45_3_331_0}
Cerioli MR, Faria L, Ferreira TO, Protti F.
\newblock A note on maximum independent sets and minimum clique partitions in
  unit disk graphs and penny graphs: complexity and approximation.
\newblock \emph{RAIRO - Theoretical Informatics and Applications - Informatique
  Th\'eorique et Applications}, 2011.
\newblock \textbf{45}(3):331--346.   doi:10.1051/ita/2011106.

\bibitem{CLARK1990165}
Clark BN, Colbourn CJ, Johnson DS.
\newblock Unit disk graphs.
\newblock \emph{Discrete Mathematics}, 1990.
\newblock \textbf{86}(1):165--177.  doi:10.1016/0012-365X(90)90358-O.

\bibitem{Altschuler2020}
Altschuler JM, Boix-Adsera E.
\newblock {Hardness results for Multimarginal Optimal Transport problems}.
\newblock 2020.
\newblock \eprint{2012.05398}.

\bibitem{DO_ALTSCHULER2021}
Altschuler JM, Boix-Adserà E.
\newblock Hardness results for Multimarginal Optimal Transport problems.
\newblock \emph{Discrete Optimization}, 2021.
\newblock \textbf{42}:100669.
\newblock \doi{https://doi.org/10.1016/j.disopt.2021.100669}.

\bibitem{buckingham1938classical}
Buckingham RA.
\newblock The classical equation of state of gaseous helium, neon and argon.
\newblock \emph{Proceedings of the Royal Society of London. Series A.
  Mathematical and Physical Sciences}, 1938.

\bibitem{pisinger1999linear}
Pisinger D.
\newblock Linear time algorithms for knapsack problems with bounded weights.
\newblock \emph{Journal of Algorithms}, 1999.
\newblock \textbf{33}(1):1--14.   doi:10.1006/jagm.1999.1034.

\bibitem{axiotis2018capacitated}
Axiotis K, Tzamos C.
\newblock Capacitated Dynamic Programming: Faster Knapsack and Graph
  Algorithms, 2018.
 \eprint{1802.06440}.

\bibitem{Hastad1999}
H{\aa}stad J.
\newblock {Clique is hard to approximate within n-e}.
\newblock \emph{Acta Mathematica}, 1999.
\newblock \textbf{182}(1):105--142.  doi:10.1007/BF02392825.

\bibitem{doi:10.1137/0132071}
Garey M, Johnson D.
\newblock The Rectilinear Steiner Tree Problem is \$NP\$-Complete.
\newblock \emph{SIAM Journal on Applied Mathematics}, 1977.
\newblock \textbf{32}(4):826--834. doi:10.1137/0132071.
\end{thebibliography}
\end{document}